\documentclass[psamsfonts,fceqn,reqno]{amsart}
\usepackage{mathrsfs,latexsym,amsfonts,amssymb,curves,epic}
\usepackage{color}
\usepackage{amsfonts, mathrsfs}
\usepackage{graphics}
\usepackage{graphicx}
\usepackage{float}
\usepackage{epstopdf}
\usepackage{epsfig}
\usepackage{amsmath, amsthm, amssymb}
\usepackage{tabularx}

\theoremstyle{definition}

\newtheorem{prop}{Proposition}

\begin{document}
\title[]
{Pricing variance swaps in a hybrid model of\\ stochastic volatility
	and interest rate\\ with regime-switching}
 
  \author{Jiling Cao}
  \address{(Jiling Cao): School of Engineering, Computer and
  Mathematical Sciences, Auckland University of Technology,
  Private Bag 92006, Auckland 1142, New Zealand}
  \email{jiling.cao@aut.ac.nz}
  
   \author{Teh Raihana Nazirah Roslan}
   \address{(Teh Raihana Nazirah Roslan): School of Engineering, Computer and
   	Mathematical Sciences, Auckland University of Technology,
   	Private Bag 92006, Auckland 1142, New Zealand\\ and
   	School of Quantitative Sciences, College of Arts
   	\& Sciences, Universiti Utara Malaysia, 06010 Sintok, Kedah,
   	Malaysia}
   \email{raihana.roslan@aut.ac.nz}
   
   \author{Wenjun Zhang}
   \address{(Wenjun Zhang): School of Engineering, Computer and
   	Mathematical Sciences, Auckland University of Technology,
   	Private Bag 92006, Auckland 1142, New Zealand}
   \email{wenjun.zhang@aut.ac.nz}


  \keywords{Heston-CIR hybrid model, Regime-switching, Realized
  variance, Stochastic interest rate, Stochastic volatility,
  Variance swap.}
  \subjclass[2000]{Primary 91G30; Secondary 91G20, 91B70.}

  \begin{abstract}
  In this paper, we consider the problem of pricing discretely-sampled
  variance swaps based on a hybrid model of stochastic volatility and
  stochastic interest rate with regime-switching. Our modeling framework
  extends the Heston stochastic volatility model by including the
  CIR stochastic interest rate and model parameters
  that switch according to a continuous-time observable Markov chain
  process. A semi-closed form pricing formula for variance swaps is
  derived. The pricing formula is assessed through numerical
  implementations, and the impact of including regime-switching on
  pricing variance swaps is also discussed.
  \end{abstract}

  \maketitle

   \section{Introduction}
  
   A \emph{variance swap} is a forward contract on the future realized
   variance of returns of a specified asset. At maturity time $T>0$,
   the variance swap rate can be evaluated as $V(T)=(RV-K)\times L$,
   where $K$ is the annualized delivery or strike price for the
   swap, $RV$ is the realized variance of the swap and $L$ is the 
   notional amount of the swap in dollars. A typical formula for 
   measuring $RV$ is
   \begin{equation}\label{vswaps}
   	RV=\frac{AF}{N}\sum_{j=1}^N \left(\frac{S(t_{j})-S(t_{j-1})}
   	{S(t_{j-1})} \right)^2 \times 100^2,
   \end{equation}
   where $S(t_{j})$ is the closing price of the underlying asset at
   the \emph{j}-th observation time $t_{j}$ and $N$ is the number
   of observations. The annualized factor $AF$ follows the sampling
   frequency to convert the above evaluation to annualized variance
   points. Assuming there are $252$ business days in a year, then
   $AF$ is equal to 252 for daily sampling frequency. However, if
   the sampling frequency is monthly or weekly, then $AF$ will be
   $12$ or $52$, respectively. The measure of realized variance
   requires monitoring the underlying price path discretely, usually
   at the end of a business day. For this purpose, we assume equally
   discrete observations to be compatible with the real market, which
   reduces to $AF=\frac{1}{\Delta t}=\frac{N}{T}$. The long position
   of variance swaps pays a fixed delivery price $K$ at the expiration
   and receives the floating amounts of annualized realized variance,
   whereas the short position is the opposite.
   
   \medskip
   Since variance swaps were first launched in 1998, the problem of
   how to price them has been an active research topic in mathematical
   and quantitative finance. Carr and Madan \cite{carr1998ttv}
   combined static replication using options with dynamic trading
   in the futures to price and hedge variance swaps without specifying
   the volatility process. Demeterfi et al.~\cite{demeterfi1999mty}
   worked in the same direction by proving that a variance swap
   could be reproduced via a portfolio of standard options. A
   finite-difference method via dimension-reduction approach was
   explored in \cite{little-pant:01} to obtain high efficiency and
   accuracy for pricing discretely-sampled variance swaps. In 
   \cite{zhu-lian:11, zhu-lian:12}, Zhu and Lian extended the work 
   in \cite{little-pant:01} by incorporating Heston two-factor 
   stochastic volatility for pricing discretely-sampled variance 
   swaps. However, a simpler approach was explored in 
   \cite{Rujivan-Zhu:12}, where the Schwarts solution procedure was 
   applied to derive an affine solution of PDEs. Recently, to extend 
   the work in \cite{zhu-lian:11} where stochastic interest rates 
   were ignored, Cao et al. \cite{Cao-Lian-Roslan} employed 
   a hybridization of the stochastic volatility model and the CIR 
   interest rate model to investigate the pricing rates of variance
   swaps with discrete sampling. In \cite{Elliott2007pvs}, Elliott
   et al. proposed a continuous-time Markovian-regulated version 
   of the Heston stochastic volatility model to distinguish different 
   states of a business cycle. An analytical formula for pricing 
   volatility swaps was obtained using the regime-switching 
   Esscher transform and comparisons were made between models with 
   and without switching regimes. The essence of incorporating 
   regime-switching for pricing variance swaps under the Heston 
   stochastic volatility model was illustrated in 
   \cite{Elliott2012pva, Elliott2007pvs}, where a common 
   assumption is ``continuous sampling time". In fact, options of 
   discretely sampled variance swaps were misvalued when the 
   continuous sampling was used as an approximation, and large 
   inaccuracies occurred in certain sampling periods, as discussed 
   in \cite{Bernard-Cui:14, Elliott2012pva, little-pant:01, 
   	zhu-lian:12}.
   
   \medskip
   In the past decade, many researchers have considered to integrate 
   Markovian regime-switching techniques with stochastic interest 
   rate models. For example, in order to incorporate jumps and 
   inconsistencies between different business stages, Elliott et al. 
   \cite{Elliott2007pvs} and Siu \cite{Siu:10} used the 
   regime-switching approach to extend the Cox-Ingersoll-Ross (CIR), 
   the Hull-White and the Vasicek models respectively. However,
   there exists a gap in the literature regarding pricing volatility 
   derivatives under stochastic volatility and stochastic interest 
   rates with regime-switching. As far as we know, the only 
   existing study was the one conducted in \cite{shen2013pvs}, 
   which focused only on continuous sampling variance swaps and 
   employed the PDE approach. In this paper, we address the issue 
   of pricing discretely-sampled variance swaps under stochastic 
   volatility and stochastic interest rate with regime-switching. 
   We extend the framework of both \cite{Cao-Lian-Roslan} and
   \cite{Elliott2012pva} by incorporating the CIR stochastic 
   interest rate into the Markov-modulated version of the Heston 
   stochastic volatility model. This hybrid model possesses 
   parameters that switch according to a continuous-time 
   observable Markov chain process which can be interpreted as 
   the states of an observable macroeconomic factor. Our approach
   is different from that of \cite{shen2013pvs}. Instead of the
   continuous sampling approach, we use the discrete sampling 
   approach to improve accuracy in pricing and computational 
   efficiency.
   
   \medskip
   The rest of this paper is organized as follows. In Section 2, a
   detailed description of regime-switching hybrid model is first
   provided, followed by derivation of the dynamics for the model
   under the \emph{T}-forward measure. In Section 3, we derive the
   forward characteristic function in order to obtain the
   semi-analytical formula for the price of variance swaps. In
   Section 4, some numerical examples are given, demonstrating
   the accuracy of our solution and impacts of regime-switching.
   In Section 5, a brief summary and comparisons of our results
   with other relevant results in the literature are provided.
   
   \section{Modelling framework}
   
   In this section, we develop a hybrid model which combines the Heston
   stochastic volatility model with the one-factor CIR stochastic
   interest rate dynamics including regime-switching effects.
   A regime-switching model for pricing volatility derivatives
   was first considered by Elliot et al. \cite{Elliott2007pvs}. Recently,
   Elliot and Lian \cite{Elliott2012pva} considered regime-switching
   effects on the Heston's stochastic volatility model. Our aim is to
   extend the work in \cite{Elliott2012pva} by incorporating stochastic
   interest rate into the modeling framework.
   
   \subsection{The Heston-CIR model with regime-switching}
   
   Let $\{S(t): 0\le t \le T\}$ be the process of certain asset price
   over a finite time horizon $[0,T]$. The Heston-CIR hybrid model is
   described by
   \begin{equation}
   \left \{\label{Eq2.1}
   \begin{array}{ll}
   dS(t)=\mu S(t)dt+\sqrt{\nu(t)}S(t)dW_{1}(t), \quad 0\le t \le T,
   \\[0.5em]
   d\nu(t)=\kappa(\theta-\nu(t))dt+\sigma\sqrt{\nu(t)}dW_{2}(t),
   \quad 0\le t \le T,\\[0.5em]
   dr(t)=\alpha(\beta-r(t))dt +\eta \sqrt{r(t)} dW_{3}(t), \quad
   0\le t \le T,
   \end{array}
   \right.
   \end{equation}
   where $\{\nu(t): 0\le t \le T\}$ is the stochastic instantaneous
   variance process and $\{r(t): 0 \le t \le T\}$ is the process of
   stochastic instantaneous interest rate. The
   parameter $\kappa$ determines the mean-reverting speed of $\nu(t)$,
   $\theta$ is its long-term mean and $\sigma$ is its volatility.
   Similarly, $\alpha$ determines the speed of mean reversion for
   the interest rate process, $\beta$ is the interest rate term
   structure and $\eta$ controls the volatility of the interest
   rate. As mentioned in \cite{cox1985tts,hestoncfs1993}, to
   ensure that the square root processes are always
   positive, it is required that $2\kappa \theta \geq \sigma^2$ and
   $2\alpha\beta \geq \eta^2$ respectively. Here, we
   assume that correlations involved in the above model are given by
   $(dW_{1}(t), dW_{2}(t))=\rho dt$, $(dW_{1}(t), dW_{3}(t))=0$ and
   $ (dW_{2}(t), dW_{3}(t))=0$, where $\rho$ is a constant with
   $-1\le \rho \le 1$. By the Girsanov theorem, there exists a
   risk-neutral measure $\mathbb Q$ equivalent to
   the real world measure $\mathbb P$ such that under
   $\mathbb Q$ system (\ref{Eq2.1}) is transformed into the form of
   \begin{equation}
   \left \{\label{Eq2.2}
   \begin{array}{ll}
   dS(t)=r(t)S(t)dt+\sqrt{\nu(t)}S(t)d\widetilde{W}_{1}(t),
   \quad 0\le t \le T,\\[0.5em]
   d\nu(t)=\kappa^*
   (\theta^*-\nu(t))dt+\sigma\sqrt{\nu(t)}d\widetilde{W}_{2}(t),
   \quad 0\le t \le T,\\[0.5em]
   dr(t)=\alpha^*(\beta^*-r(t))dt +\eta \sqrt{r(t)} d\widetilde{W}_{3}(t),
   \quad 0\le t \le T,
   \end{array}
   \right.
   \end{equation}
   where $\kappa^* = \kappa+\lambda_1$, $\theta^* =\frac{\kappa\theta}
   {\kappa+\lambda_1}$, $\alpha^*=\alpha+\lambda_2$ and $\beta^*=
   \frac{\alpha\beta}{\alpha+\lambda_2}$ are the risk-neutral
   parameters, $\{\widetilde{W}_{i}(t): 0\le t \le T\}$
   ($1\le i\le 3$) is a Brownian motion under $\mathbb Q$. Here,
   $\lambda_j$ ($j=1,2$) is the premium of volatility or interest
   rate risk.
   
   \medskip
   The market dynamics is modelled by a continuous-time observable
   Markov chain ${\bf X}=\{X(t): 0\le t \le T \}$ with a finite
   state space $S=\{s_1,s_2,...,s_N\}$. Without loss of generality,
   $S$ can be identified with the set of unit vectors $\{e_1,
   e_2,..., e_N\}$, where $e_i=(0,...,1,...,0)^{\intercal} \in
   \mathbb {R}^N$. An $N$-by-$N$ rate matrix $Q=(q_{ij})_{1 \leq
   	i,j \leq N}$ is used to generate the evolution of the chain
   under $\mathbb Q$. Here, $q_{ij} \geq 0$ for all $1\le i, j
   \le N$ with $i \neq j$ and $\sum_{i=1}^{N} q_{ij}=0$ for all
   $1 \le j \le N$. According to \cite{Elliott-Wilson:07}, a
   semi-martingale representation holds for the process ${\bf X}$
   as follows
   \begin{equation}\label{dynamicX2}
   X(t)= X(0) + \int_0^t {QX(s) ds} + M(t),
   \end{equation}
   where $\{M(t):0 \leq t \leq T \}$ is a $\mathbb R^{N}$-valued
   martingale with respect to the filtration generated by $\bf X$
   under $\mathbb Q$. The regime-switching effect is captivated
   in our Heston-CIR model by assuming that the asset price, its
   volatility and the interest rate depend on market trends or
   other economic factors indicated by the regime-switching
   Markov chain ${\bf X}$. More precisely, the long-term mean of
   variance $\theta^*(t)$ of the asset price is given by
   $\theta^*(t)=\langle \theta^*,X(t) \rangle$, where $\theta^*=
   (\theta^*_1,\theta^*_2,...,\theta^*_N)^{\intercal}$
   with $\theta^*_i > 0$, for each $1 \le i \le N$, and $\langle
   \cdot, \cdot \rangle$ denotes the scalar product in $\mathbb
   R^{N}$. Similarly, the long-term mean of
   the interest rate $\beta^*(t)$ is given by $\beta^*(t)=\langle
   \beta^*,X(t) \rangle$, where $\beta^* = (\beta^*_1, \beta^*_2,
   ..., \beta^*_N)^{\intercal}$ with $\beta^*_i > 0$, for each
   $1 \le i \le N$. The Heston-CIR model under $\mathbb Q$ with
   regime switching is given by
   \begin{equation}\label{riskneutralRS}
   \left \{
   \begin{array}{ll}
   dS(t)=r(t)S(t)dt+\sqrt{\nu(t)}S(t)d\widetilde{W}_{1}(t),
   \quad 0\le t \le T,\\[0.5em]
   d\nu(t)=\kappa^*
   (\theta^*(t)-\nu(t))dt+\sigma\sqrt{\nu(t)}d\widetilde{W}_{2}(t),
   \quad 0\le t \le T,\\[0.5em]
   dr(t)=\alpha^*(\beta^*(t)-r(t))dt +\eta \sqrt{r(t)}
   d\widetilde{W}_{3}(t), \quad 0\le t \le T.
   \end{array}
   \right.
   \end{equation}
   Applying the Cholesky decomposition, we can re-write SDEs
   (\ref{riskneutralRS}) as
   \begin{equation}\label{measuretower1}
   \left(\begin{array}{cc} \frac {dS(t)}{S(t)} \\[0.5em]
   d\nu(t)\\ dr(t)\\ \end{array} \right)
   =\mu^{\mathbb Q} dt
   +\Sigma \times C \times \left(\begin{array}{cc} d W^*_{1}(t)
   \\[0.5em]
   d W^*_{2}(t) \\[0.5em]
   d W^*_{3}(t)\end{array} \right), \quad 0\le t \le T,
   \end{equation}
   with
   \begin{equation}\label{decomposechap5}
   \mu^{\mathbb Q}=\left( \begin{array}{cc} r(t) \\
   \kappa^*(\theta^*(t)-\nu(t))\\
   \alpha^*(\beta^*(t)-r(t))\\ \end{array} \right), \quad
   \Sigma=\left(
   \begin{array}{ccc} \sqrt{\nu(t)} & 0 & 0 \\ 0 &
   \sigma\sqrt{\nu(t)} & 0 \\ 0 & 0 & \eta \sqrt{r(t)}
   \end{array}
   \right)\\
   \end{equation}
   and
   \begin{equation*}
   C =\left(\begin{array}{ccc} 1 & 0 & 0 \\ \rho & \sqrt{1-\rho^2}
   & 0 \\ 0 & 0 & 1 \end{array} \right)
   \end{equation*}
   such that
   \begin{equation*}
   C C ^{\intercal}=\left(\begin{array}{ccc} 1 & \rho & 0 \\ \rho & 1 & 0 \\ 0 & 0 & 1 \end{array} \right)
   \end{equation*}
   and $dW^*_{1}(t)$, $dW^*_{2}(t)$ and $dW^*_{3}(t)$ are mutually
   independent under $\mathbb Q$ satisfying
   \begin{equation*}
   \left(\begin{array}{cc} d \widetilde{W}_{1}(t)
   \\[0.5em]
   d \widetilde{W}_{2}(t) \\[0.5em]
   d \widetilde{W}_{3}(t)\end{array} \right)
   = C \times
   \left( \begin{array}{cc} dW^*_{1}(t)
   \\[0.5em]
   dW^*_{2}(t) \\[0.5em]
   dW^*_{3}(t) \end{array} \right), \quad 0\le t \le T.
   \end{equation*}
   
   \subsection{Model dynamics under $T$-forward measure}
   
   In this subsection, we convert dynamics of the Heston-CIR model with
   regime-switching under $\mathbb Q$ to one under the $T$-forward
   measure $\mathbb Q^T$. To this end, we first derive a
   regime-switching exponential affine form for the price
   $P(t,T,r(t),X(t))$ of a zero-coupon bond under $\mathbb Q$.
   
   \medskip
   Assume that the bond price $P(t,T,r(t),X(t))$ under $\mathbb Q$ has
   the following exponential affine form
   \begin{equation}
   P(t,T,r(t),X(t))=e^{A(t,T,X(t))-B(t,T)r(t)},
   \end{equation}
   where $A(t,T,X(t))$ and $B(t,T)$ are to be determined. The discounted
   bond price is given by
   \begin{equation}
   \widetilde{P}(t,T,r(t),X(t))=e^{-\int_{0}^{t}r(s)ds} P(t,T,r(t),X(t)).
   \end{equation}
   Applying It$\hat{\rm o}$'s formula to $\widetilde{P}(t,T,r(t),X(t))$
   and noting that the non-martingale terms must sum up to zero, we
   obtain
   \begin{equation}\label{ito}
   \begin{array}{ll}
   \dfrac{\partial P}{\partial t} + \alpha^*(\beta^{*}(t)-r)
   \dfrac{\partial P}{\partial r} + \langle \textbf{P},QX(t)\rangle + \dfrac{1}{2}\dfrac{\partial^2 P}{\partial r^2}\eta^{2}r-rP=0,
   \end{array}
   \end{equation}
   with terminal condition $P(T,T,r(T),X(T))=1$, $\textbf{P}=(P_1,P_2,...,
   P_N)^{\intercal}$ and $P_{i}=P(t,T,r,e_{i})$ for  $1 \le i \le N$.
   Note that $X(t)$ takes one of the values from the set of unit vectors
   $\{e_1, e_2,..., e_N\}$. If $X(t)=e_{i}$ for some $1 \le i \le N$, then
   \begin{equation*}
   \begin{array}{ll}
   \theta^*(t)=\langle \theta^*,X(t)\rangle =\theta^*_{i}, \\[0.5em]
   \beta^*(t)=\langle \beta^*,X(t)\rangle =\beta^*_{i},\\[0.5em]
   P(t,T,r(t),X(t))=P(t,T,r(t),e_{i})=P_{i}.
   \end{array}
   \end{equation*}
   As a result, equation (\ref{ito}) becomes $N$ coupled PDEs
   \begin{equation}
   \begin{array}{ll}
   \dfrac{\partial P_{i}}{\partial t} + \alpha^*(\beta^{*}_{i}-r)
   \dfrac{\partial P_{i}}{\partial r} + \langle \textbf{P},Qe_{i}\rangle + \dfrac{1}{2}\dfrac{\partial^2 P_{i}}{\partial r^2}\eta^{2}r-rP_{i}=0,
   \quad 1\le i \le N
   \end{array}
   \end{equation}
   with terminal conditions $P_{i}(T,T,r(T))=1$. We then substitute the
   expressions of $\dfrac{\partial P}{\partial t}$, $\dfrac{\partial P}
   {\partial r}$ and $\dfrac{\partial^2 P}{\partial r^2}$ into the above
   PDEs to obtain the following ordinary differential equations as
   \begin{equation}\label{rsode}
   \left \{
   \begin{array}{ll}
   \dfrac{dB(t,T)}{dt}=\dfrac{1}{2}\eta^{2} B(t,T)^{2}+ \alpha^*B(t,T)-1,\\[0.8em]
   \dfrac{dA_{i}}{dt}=\alpha^*\beta^*_{i}B(t,T)-e^{-A_{i}} \langle \widetilde{\textbf{A}},Qe_{i}\rangle, \quad 1\le i \le N, \\[0.8em]
   \end{array}
   \right.
   \end{equation}
   where $A_{i}=A(t,T,e_{i})$, $\widetilde{A}_{i}=e^{A_{i}}$ and
   $\widetilde{\textbf{A}}=(\widetilde{A}_1,\widetilde{A}_2,...,
   \widetilde{A}_N)^{\intercal}$, $1\le i \le N$. The terminal
   conditions become $B(T,T)=0$ and $A_{i}(T,T)=0$.
   Similar to the CIR model in \cite{cox1985tts}, solution to the first
   equation of (\ref{rsode}) is
   \begin{equation*}
   \begin{array}{ll}
   B(t,T )=\dfrac{2\left(e^{(T-t)\sqrt{(\alpha^*)^2+2\eta^2}}-1 \right)}
   {2\sqrt{(\alpha^*)^2 +2\eta^2}+
   	\left(\alpha^*+\sqrt{(\alpha^*)^2+2\eta^2}\right)
   	\left(e^{(T-t)\sqrt{(\alpha^*)^2+2\eta^2}}-1\right)}.
   \end{array}
   \end{equation*}
   To derive an expression for $A_{i}$'s and $A(t,T,X(t))$, let
   $\Upsilon_{i}(t)=\alpha^* \beta^*_{i}B(t,T)$ for each $1 \le i \le N$,
   and let $\operatorname{diag}(\Upsilon(t))$ denote the diagonal matrix
   whose entry on the $i$-th row and the $i$-column is $\Upsilon_{i}(t)$
   for all $1\le i \le N$.
   Substituting $\widetilde{A}_{i}=e^{A_{i}}$ into (\ref{rsode}), we can
   re-write the system of ordinary differential equations in (\ref{rsode})
   as the following matrix form
   \begin{equation}\label{system}
   \begin{array}{ll}
   \dfrac{d\widetilde{\textbf{A}}}{dt}=\left(\operatorname{diag}
   (\Upsilon(t))-Q^{\intercal}\right) \widetilde{\textbf{A}}
   \end{array}
   \end{equation}
   with $\widetilde{\textbf{A}}(T,T)=\textbf{1}$, where $\textbf{1}=(1,1,...,
   1)^{\intercal} \in \mathbb {R}^N$. Let $\Phi(t)$ be the fundamental
   matrix of (\ref{system}) with $\Phi(T)= I_N$, where $I_N$ denotes the
   $N$-dimensional identity matrix. Then the solution to (\ref{system})
   with terminal condition $\widetilde{\textbf{A}}(T,T)=\textbf{1}$ can be
   expressed as $\widetilde{\textbf{A}}(t,T)=\Phi(t)\textbf{1}$. It follows
   that
   \begin{equation*}
   \widetilde{A}_{i}(t,T)= \langle \Phi(t)
   \textbf{1},e_{i} \rangle \quad \mbox{and} \quad
   A(t,T,X(t))=\ln (\langle \Phi(t)\textbf{1},X(t) \rangle).
   \end{equation*}
   
   \medskip
   Now, we implement the techniques of change of measure from
   $\mathbb Q$ to ${\mathbb Q}^T$. For brevity, let us denote
   the numeraire $e^{\int_{0}^{t}r(s)ds}$ by $N_{1,t}$ and the
   numeraire $P(t,T,r(t),X(t))$ by $N_{2,t}$. Then,
   \begin{equation*}
   \begin{array}{ll}
   \displaystyle
   d\ln N_{1,t}=r(t)dt=\left(\int_{0}^t \alpha^*(\beta^*-r(s))ds
   \right) dt +\left(\int_{0}^t
   \eta \sqrt{r(s)}d\widetilde{W}_{3}(s)\right) dt.
   \end{array}
   \end{equation*}
   So, the volatility for the numeraire $N_{1,t}$ is given by
   $\Sigma^{\mathbb Q}=(0,0,0)^{\intercal}$. Similarly,
   differentiating $\ln N_{2,t}=\ln\widetilde{A}(t,T,X(t)){-B(t,T)
   	r(t)}$ gives
   \begin{equation*}
   \begin{array}{ll}
   d\ln N_{2,t}= \left(\dfrac{\dfrac{\partial
   		\widetilde{A}(t,T,X(t))}{\partial t}}{\widetilde{A}(t,T,X(t))}
   -\dfrac{\partial B(t,T)}{\partial t}r(t) -B(t,T)
   \alpha^*(\beta^*(t)-r(t)) \right.\\[0.8em]
   \quad \quad \quad \quad \quad \left. +\Big\langle \dfrac{
   	\widetilde{\textbf{A}}(t,T)}{ \widetilde{A}(t,T,X(t))},QX(t)
   \Big\rangle \right)dt-B(t,T) \eta \sqrt{r(t)}
   d\widetilde{W}_{3}(t)\\[0.8em] \quad \quad \quad \quad \quad+
   \Big\langle \dfrac{\widetilde{\textbf{A}}(t,T)}
   {\widetilde{A}(t,T,X(t))},dM(t) \Big\rangle.
   \end{array}
   \end{equation*}
   Note that $d \widetilde{W}_{3}(t)$ and $dM(t)$ are independent.
   So, the volatility for the numeraire $N_{2,t}$ is given by
   $\Sigma^{T}=\left(0, 0, -B(t,T)\eta\sqrt{r(t)}\right)^{\intercal}$.
   
   \medskip
   Using a formula in \cite{brigo2006irm}, we see that the drift
   $\mu^{T}$ of our SDEs under $\mathbb{Q}^{T}$ with
   regime-switching is given by
   \begin{equation*}
   \mu^{T}= \mu^{\mathbb Q} - \left(\Sigma \times C \times
   C^{\intercal} \times (\Sigma^{\mathbb Q} -\Sigma^{T})\right)
   =\left(\begin{array}{cc} r(t)\\ \kappa^*(\theta^*(t)-\nu(t))\\ \alpha^*\beta^*(t)-[\alpha^*+B(t,T) \eta^2]r(t)\\ \end{array}
   \right)
   \end{equation*}
   with $\Sigma$ and $CC^{\intercal}$ as defined in (\ref{decomposechap5}).
   Therefore, the dynamics for (\ref{riskneutralRS}) under
   $\mathbb{Q}^T$ is given by
   \begin{equation}\label{ForwardMeasureSDE}
   \left(\begin{array}{cc} \frac {dS(t)}{S(t)} \\ d\nu(t)\\ dr(t)\\
   \end{array} \right)
   = \mu^{T} dt +\Sigma \times C \times \left(\begin{array}{cc} dW_{1}^*(t)\\dW_{2}^*(t)\\dW_{3}^*(t)\end{array} \right).
   \end{equation}
   In addition, under $\mathbb{Q}^T$, the semi-martingale decomposition
   of $\bf X$ is given by
   \begin{equation}\label{dynamicX}
   X(t)= X(0) + \int_0^t {Q^{T}(s) X(s) ds} + M^{T}(t),
   \end{equation}
   with the rate matrix $Q^{T}(t)=(q_{ij}^{T}(t))_{1\le i,j \le N}$
   defined by
   \begin{equation*}
   {q_{ij}}^{T}(t)= \left\{
   \begin{array}{ll}
   q_{ij}\dfrac{\widetilde{A}(t,T,e_j)}{\widetilde{A}(t,T,e_i)},
   & \mbox{$i \neq j$}, \\[0.8em]
   -\sum_{k \neq i} q_{ik}\dfrac{\widetilde{A}(t,T,e_k)}
   {\widetilde{A}(t,T,e_i)}, & \mbox{$i = j$},
   \end{array}
   \right.
   \end{equation*}
   refer to \cite{palmowski-rolski:2002} for details.
   \bigskip
   \section{Derivation of pricing formula}
   
   In this section, we will derive a semi-closed form solution to the
   problem of pricing variance swaps under stochastic volatility and
   stochastic interest rate with regime-switching using characteristic
   functions. Let $y(T)=\ln S(T+\Delta)-\ln S(T)$. We have to evaluate
   the price conditional on the information about the sample path
   of $\bf X$ from $t=0$ to $t=T+\Delta$.
   First, define ${\mathscr F}_1(t)$, ${\mathscr F}_2(t)$ and ${\mathscr
   	F}_3(t)$ as the natural filtrations generated by $\{W_1^*(s): 0\le s
   \le t\}$, $\{W_2^*(s): 0\le s \le t\}$ and $\{W_3^*(s): 0\le s \le t\}$,
   respectively. Let ${\mathscr F}_X(t)$ be the filtration generated by
   $\{X(s): 0\le s \le t\}$. To obtain the characteristic function of
   $y(T)$, we need to evaluate the following conditional expectation
   in two steps:
   
   \begin{equation}\label{regime-tower}
   \begin{array}{ll}
   \mathbb{E}^T(e^{\phi y(T)}|\mathscr F_1(t) \vee \mathscr
   F_2(t) \vee \mathscr F_3(t) \vee \mathscr F_X(t))\\[0.5em]
   =\mathbb{E}^T(\mathbb{E}^T(e^{\phi y(T)}|\mathscr F_1(t) \vee
   \mathscr F_2(t) \vee \mathscr F_3(t) \vee \mathscr F_X(T+\Delta))\\[0.5em]
   \quad |\mathscr F_1(t) \vee \mathscr F_2(t) \vee \mathscr F_3(t) \vee
   \mathscr F_X(t))
   \end{array}
   \end{equation}
   In the first step, we compute
   $\mathbb{E}^{T} \left(e^{\phi y(T)}|\mathscr F_1(t) \vee \mathscr F_2(t)
   \vee \mathscr F_3(t) \vee \mathscr F_X(T+\Delta)\right)$. In the second
   step, we compute
   $\mathbb{E}^{T}(e^{\phi y(T)}|\mathscr F_1(t) \vee \mathscr F_2(t)
   \vee \mathscr F_3(t) \vee \mathscr F_X(t))$.
   
   \subsection{Characteristic function for the given path $\mathscr
   	F_X(T+\Delta)$}
   
   We consider an enlarged filtration in which the forward characteristic
   function $f(\phi; t,T,\Delta,\nu(t),r(t))$ of $y(T)$ is defined by
   \begin{equation}
   \begin{array}{ll}
   f(\phi; t,T,\Delta,\nu(t),r(t))
   =\mathbb{E}^{T}(e^{\phi y(T)}|\mathscr F_1(t) \vee \mathscr F_2(t)
   \vee \mathscr F_3(t) \vee \mathscr F_X(t))
   \end{array}
   \end{equation}
   
   \begin{prop}
   	If the underlying asset follows the dynamics (\ref{ForwardMeasureSDE}),
   	then
   	\begin{eqnarray}\label{enlargedfiltration}
   	&&  f(\phi; t,T,\Delta,\nu(t),r(t)|\mathscr F_X(T+\Delta))\nonumber\\
   	&=&e^{C(\phi,T)} j(D(\phi,T);t,T,\nu(t)) \cdot k(E(\phi,T);t,T,r(t)),
   	\nonumber
   	\end{eqnarray}
   	where for any $0\le t \le T$, $D(\phi,t)$, $j(\phi; t,T,\nu(t))$ and
   	$k(\phi; t,T,r(t))$ are given by
   	\begin{equation}
   	\begin{array}{ll}
   	D(\phi,t)=\dfrac{a+b}{\sigma^2}\dfrac{1-e^{b(T+\Delta-t)}}
   	{1-ge^{b(T+\Delta-t)}},\\[0.8em]
   	a=\kappa^*-\rho\sigma\phi, \quad
   	b=\sqrt{a^2+\sigma^2(\phi-\phi^2)}, \quad g=\dfrac{a+b}{a-b},\\[0.8em]
   	\end{array}
   	\end{equation}
   	with
   	\begin{equation*}
   	\begin{array}{ll}
   	j(\phi;t,T,\nu(t))=e^{F(\phi,t)+G(\phi,t)\nu(t)}, \quad
   	F(\phi,t)=\int_t^{T} \langle \kappa^{*}\theta^{*}G(\phi,s),X(s)
   	\rangle ds,\\[0.8em]
   	G(\phi,t)=\dfrac{2 \kappa^{*}\phi}{\sigma^2 \phi+(2\kappa^{*}-\sigma^2\phi)e^{\kappa^{*}(T-t)}}, \quad
   	k(\phi;t,T,r(t))=e^{L(\phi,t)+M(\phi,t)r(t)},\\[0.8em]
   	\end{array}
   	\end{equation*}
   	and $C(\phi,t)$, $E(\phi,t)$, $L(\phi,t)$ and $M(\phi,t)$ are determined by
   	the following ODEs
   	\begin{equation}
   	\left \{
   	\begin{array}{ll}
   	-\dfrac{dE}{dt}=\dfrac{1}{2} \eta^2 E^2 -(\alpha^*+B(t,T)\eta^2) E +\phi, \\[0.8em]
   	-\dfrac{dC}{dt}=\kappa^*\theta^*(t) D+ \alpha^* \beta^*(t) E, \\[0.8em]
   	-\dfrac{dM}{dt}=\dfrac{1}{2} \eta^2 M^2 -(\alpha^*+B(t,T)\eta^2)M, \\[0.8em]
   	-\dfrac{dL}{dt}=\alpha^* \beta^*(t)M.
   	\end{array}
   	\right.
   	\end{equation}
   \end{prop}
   
   \begin{proof}
   	Here, we give a brief proof for Proposition 1. We represent the
   	conditional forward characteristic function for $y(T)$ as
   	\begin{eqnarray}\label{chacfxofy}
   	&& f(\phi;t,T,\Delta,\nu(t),r(t)|\mathscr F_X(T+\Delta)) \nonumber \\
   	&=& \mathbb{E}^T(\mathbb{E}^T(e^{\phi y(T)}|\mathscr F_1(T) \vee
   	\mathscr F_2(T) \vee \mathscr F_3(T)\vee \mathscr F_X(T+\Delta))\\
   	&& |\mathscr F_1(t) \vee \mathscr F_2(t) \vee \mathscr F_3(t)\vee \mathscr
   	F_X(T+\Delta)). \nonumber
   	\end{eqnarray}
   	We first focus on calculating the inner expectation
   	\begin{equation*}
   	\mathbb{E}^{T}(e^{\phi y(T)}|\mathscr F_1(T) \vee \mathscr F_2(T) \vee \mathscr
   	F_3(T)\vee \mathscr F_X(T+\Delta)).
   	\end{equation*}
   	By defining function
   	\begin{equation*}
   	U(\phi;t,\widetilde{s},\nu,r)=\mathbb{E}^T(e^{\phi y(T)}|\mathscr F_1(t)
   	\vee \mathscr F_2(t) \vee \mathscr F_3(t) \vee \mathscr F_X(T+\Delta))
   	\end{equation*}
   	with $T \leq t \leq T+\Delta$, and applying the Feynman-Kac theorem, we obtain
   	\begin{equation}\label{eqn:pdeinner}
   	\begin{array}{ll}
   	\dfrac{\partial U}{\partial t}+\dfrac{1}{2}\nu\dfrac{\partial^2 U}{\partial
   		\widetilde{s}^2} + \dfrac{1}{2}\sigma^2\nu\dfrac{\partial^2 U}
   	{\partial \nu^2} + \dfrac{1}{2}\eta^2 r \dfrac{\partial^2 U}
   	{\partial r^2}+ \rho\sigma \nu \dfrac{\partial^2 U}
   	{\partial \widetilde{s}\partial \nu} +\left(r -\dfrac{1}{2}\nu \right)
   	\dfrac{\partial U}{\partial \widetilde{s}}\\[0.8em]
   	+\left(\kappa^*(\theta^*(t)-\nu)\right) \dfrac{\partial U}{\partial \nu} + \left(\alpha^*\beta^*(t)-(\alpha^*+B(t,T)\eta^2)r \right)
   	\dfrac{\partial U}{\partial r}=0,\\[0.8em]
   	U(\phi;t=T+\Delta,\widetilde{s},\nu,r)=e^{\phi y(T)},
   	\end{array}
   	\end{equation}
   	where $\widetilde{s}(t)=\ln S(t)-\ln S(T)$ in $T \leq t \leq T+\Delta$.
   	In order to solve (\ref{eqn:pdeinner}), we assume $U(\phi;t,\widetilde{s},
   	\nu(t),r(t))$  in \cite{hestoncfs1993} has the following affine form
   	\begin{equation}\label{eqn:affineinner}
   	U(\phi;t,\widetilde{s},\nu(t),r(t))=e^{C(\phi,t)+D(\phi,t)\nu +E(\phi,t)r
   		+\phi \widetilde{s}}.
   	\end{equation}
   	Substituting (\ref{eqn:affineinner}) into (\ref{eqn:pdeinner}), we obtain
   	the following three ODEs
   	\begin{equation} \label{eqn:odestep1}
   	\left \{
   	\begin{array}{ll}
   	-\dfrac{dD}{dt}=\dfrac{1}{2}\phi(\phi-1)+(\rho\sigma\phi-\kappa^{*})D +
   	\dfrac{1}{2}\sigma^2 D^2,\\[0.8em]
   	-\dfrac{dE}{dt}=\dfrac{1}{2} \eta^2 E^2 -(\alpha^*+B(t,T)\eta^2) E +\phi, \\[0.8em]
   	-\dfrac{dC}{dt}=\kappa^*\theta^*(t) D+ \alpha^* \beta^*(t) E, \\[0.8em]
   	\end{array}
   	\right.
   	\end{equation}
   	with the initial conditions
   	\begin{equation}
   	C(\phi,T+\Delta)=0, \quad D(\phi,T+\Delta)=0, \quad E(\phi,T+\Delta)=0.
   	\end{equation}
   	Then, we can write the solution to the first ODE in (\ref{eqn:odestep1})
   	as
   	\begin{equation}
   	\left \{
   	\begin{array}{ll}
   	D(\phi,t)=\dfrac{a+b}{\sigma^2}\dfrac{1-e^{b(T+\Delta-t)}}
   	{1-ge^{b(T+\Delta-t)}},\\[0.8em]
   	a=\kappa^*-\rho\sigma\phi, \quad
   	b=\sqrt{a^2+\sigma^2(\phi-\phi^2)}, \quad g=\dfrac{a+b}{a-b}.\\[0.8em]
   	\end{array}
   	\right.
   	\end{equation}
   	Numerical integration is required to obtain the solutions of $E$ and $C$.
   	
   	\medskip
   	Now, we move on to solve the outer expectation for $0 \leq t \leq T$. At
   	$t=T$,
   	\begin{equation*}
   	\begin{array}{ll}
   	\quad \mathbb{E}^T(e^{\phi y(T)}|\mathscr F_1(T) \vee \mathscr F_2(T)
   	\vee \mathscr F_3(T)\vee \mathscr F_{X}(T+\Delta))\\[0.8em]
   	=U(\phi;t=T,\widetilde{s}(T),\nu(T),r(T))\\[0.8em]
   	=e^{C(\phi,T)+D(\phi,T)\nu(T) +E(\phi,T)r(T)}.
   	\end{array}
   	\end{equation*}
   	Define the following characteristic functions of $\nu(t)$ and $r(t)$,
   	respectively
   	\begin{equation*}
   	j(\phi;t,T,\nu(t))=\mathbb{E}^{T}(e^{\phi \nu(T)}|\mathscr F_1(t) \vee
   	\mathscr F_2(t) \vee \mathscr F_3(t)),
   	\end{equation*}
   	and
   	\begin{equation*}
   	k(\phi;t,T,r(t))=\mathbb{E}^{T}(e^{\phi r(T)}|\mathscr F_1(t) \vee
   	\mathscr F_2(t) \vee \mathscr F_3(t)).
   	\end{equation*}
   	Then, we obtain the respective PDEs as
   	\begin{equation}
   	\left\{
   	\begin{array}{ll}\label{functionj}
   	\dfrac{\partial j}{\partial t}+\dfrac{1}{2}\sigma^2 \nu \dfrac{\partial^2 j}
   	{\partial \nu^2}+\left(\kappa^*(\theta^*(t)-\nu)\right)\dfrac{\partial j}{\partial \nu}=0,\\[0.8em]
   	j(\phi,t=T,T,\nu)=e^{\phi \nu},
   	\end{array}
   	\right.
   	\end{equation}
   	and
   	\begin{equation}\label{PDEk}
   	\left\{
   	\begin{array}{ll}
   	\dfrac{\partial k}{\partial t}+\dfrac{1}{2}\eta^2 r \dfrac{\partial^2 k}
   	{\partial r^2}+\left(\alpha^*\beta^*(t)-(\alpha^*+B(t,T)\eta^2)r \right)
   	\dfrac{\partial k}{\partial r}=0,\\[0.8em]
   	k(\phi,t=T,T,r)=e^{\phi r}.
   	\end{array}
   	\right.
   	\end{equation}
   	Taking advantage of the affine-form solution techniques as those in \cite{Duffie-Pan-Singleton:00,hestoncfs1993},
   	we assume the solution to (\ref{functionj}) is in the form of
   	\begin{equation}
   	j(\phi;t,T,\nu(t))=e^{F(\phi,t)+G(\phi,t)\nu(t)}.
   	\end{equation}
   	The functions $F(\phi,t)$ and $G(\phi,t)$ can be found by solving two ODEs
   	\begin{equation}
   	\left\{
   	\begin{array}{ll}
   	-\dfrac{dG}{dt}=\dfrac{1}{2} \sigma^2 G^2 -\kappa^{*}G, \\[0.8em]
   	-\dfrac{dF}{dt}=\kappa^* \theta^{*}(t)G, \\
   	\end{array}
   	\right.
   	\end{equation}
   	with the initial conditions
   	\begin{equation}
   	F(\phi,T)= 0,  \quad \quad G(\phi,T)=\phi.
   	\end{equation}
   	The solutions are
   	\[
   	F(\phi,t) =\int_t^T \kappa^{*}\theta^{*}(s)G(\phi,s)ds,\quad
   	G(\phi,t)= \dfrac{2 \kappa^{*}\phi}{\sigma^2 \phi+(2\kappa^{*}-\sigma^2\phi)e^{\kappa^{*}(T-t)}}.
   	\]
   	Next, the function $k(\phi;t,T,r(t)) = e^{L(\phi,T)+M(\phi,t)r(t)}$ is defined
   	in order to derive a solution to (\ref{PDEk}). The initial conditions are
   	$L(\phi,T)= 0$ and $M(\phi,T)=\phi$. Then, $L$ and $M$ satisfy the following ODEs
   	\begin{equation}
   	\left \{
   	\begin{array}{ll}
   	-\dfrac{dM}{dt}=\dfrac{1}{2} \eta^2 M^2 -(\alpha^*+B(t,T)\eta^2)M, \\[0.8em]
   	-\dfrac{dL}{dt}=\alpha^* \beta^*(t)M,
   	\end{array}
   	\right.
   	\end{equation}
   	Combining the inner and outer expectation computations, we obtain
   	the result claimed in the proposition.
   \end{proof}
   
   \subsection{Characteristic function for the given path $\mathscr F_X(t)$}
   
   In this subsection, we derive a semi-closed formula for the characteristic
   function $f(\phi; t,T,\Delta,\nu(t),r(t))$. To achieve this, we need to
   evaluate the equation (\ref{enlargedfiltration}), where $\theta^{*}(t)$ and
   $\beta^{*}(t)$ depend on the path of the Markov chain process $\bf X$ up to
   $T+\Delta$,
   \begin{eqnarray}\label{core}
   && f(\phi; t,T,\Delta,\nu(t),r(t)) \nonumber\\[0.8em]
   &=&\mathbb{E}^T(e^{C(\phi,T)} \cdot j(D(\phi,T);t,T,\nu(t)) \cdot
   k(E(\phi,T);t,T,r(t)) \nonumber \\
   && |\mathscr F_1(t) \vee \mathscr F_2(t) \vee \mathscr F_3(t)
   \vee \mathscr F_X(t)) \nonumber \\
   &=&\mathbb{E}^{T}\left(\exp \left(\int_T^{T+\Delta} {\langle \alpha^{*}
   	\beta^{*}E(\phi,s) + \kappa^{*}\theta^{*}D(\phi,s), X(s) \rangle} ds \right.\right.\nonumber\\[0.8em]
   && +\int_t^{T} \langle \kappa^{*}\theta^{*}G(D(\phi,T),s),X(s) \rangle ds
   +\int_t^{T} \langle \alpha^{*}\beta^{*}M(E(\phi,T),s),X(s) \rangle ds
   \nonumber\\[0.8em]
   && +\dfrac{2 \kappa^{*}D(\phi,T)}{\sigma^2 D(\phi,T)+(2\kappa^{*}-\sigma^2
   	D(\phi,T)) e^{\kappa^{*}(T-t)}} \nu(t)  \\[0.5em]
   && \left.\quad+r(t)\int_t^T {\left(\dfrac{1}{2} \eta^2 M^2(E(\phi,T),s)-(\alpha^*+B(s,T)\eta^2)M(E(\phi,T),s) \right)}ds \right)
   \nonumber \\[0.5em]
   && \left.\quad\Big|\mathscr F_1(t) \vee \mathscr F_2(t) \vee
   \mathscr F_3(t) \vee \mathscr F_X(t)\right) \nonumber\\[0.8em]
   &=&\mathbb{E}^T\left(\exp(\int_t^{T+\Delta} \langle J(s),X(s) \rangle ds)
   |\mathscr F_1(t) \vee \mathscr F_2(t) \vee \mathscr F_3(t)
   \vee \mathscr F_X(t)\right) \nonumber \\[0.8em]
   && \quad\times \exp\left(\nu(t)G(D(\phi,T),t)\right)
   \times \exp\left(r(t)M(E(\phi,T),t)\right). \nonumber
   \end{eqnarray}
   Here, the function $J(t) \in \mathbb{R}^{N}$ is given by
   \begin{equation}
   \begin{array}{ll}
   J(t)=(\kappa^{*}\theta^{*}G(D(\phi,T),t)+\alpha^{*}\beta^{*}M(E(\phi,T),t)) (1-H_{T}(t))\\[0.8em]
   \quad \quad \quad + (\alpha^{*}\beta^{*}E(\phi,t)+\kappa^{*}\theta^{*}
   D(\phi,t))H_{T}(t)
   \end{array}
   \end{equation}
   along with $H_{T}(t)$ which is a Heaviside unit step function defined as
   \begin{equation*}
   H_T(t)=\left\{\begin{array}{ll}
   1, & \mbox{if $t \geq T$},\\
   0, & \mbox{else}.\end{array}\right.
   \end{equation*}
   
   \begin{prop}
   	Let $\{X(t): 0\le t \le T\}$ be a regime-switching Markov chain with
   	dynamics given by (\ref{dynamicX}). Under ${\mathbb Q}^T$,
   	$\exp \left(\int^{T}_t {\langle J(s),X(s) \rangle  ds} \right)$ is
   	given by
   	\begin{equation} \label{eqn:tforward}
   	\begin{array}{ll}
   	\quad \mathbb{E}^T\left(\exp \left(\int^{T}_t {\langle J(s),X(s) \rangle
   		ds} \right) \Big|\mathscr F_1(t) \vee \mathscr F_2(t) \vee \mathscr F_3(t)
   	\vee \mathscr F_X(t) \right)\\[0.8em]
   	=\langle \Phi(t,T;J)X(t), \mathbf{1} \rangle,
   	\end{array}
   	\end{equation}
   	where the function $\Phi(t,T;J)$ is an $N$-by-$N$ $\mathbb{R}$-valued matrix
   	given by
   	\begin{equation}
   	\Phi(t,T;J)=\exp \left( \int^{T}_t {(Q^{T}(s) + \operatorname{diag}(J(s)))
   		ds} \right),
   	\end{equation}
   	with $\mathbf{1}=(1,1,...,1) \in \mathbb{R}^{N}$.
   \end{prop}
   
   \begin{proof}
   	Consider $Z(t,T)=\exp \left(\int^{T}_t {\langle J(s),X(s) \rangle  ds} \right)
   	X(T)$. Differentiating $Z(t,T)$ and using (\ref{dynamicX}) yield
   	\begin{equation}
   	\begin{array}{ll}\label{diffZ}
   	dZ(t,T)=\exp \left(\int^{T}_t {\langle J(s),X(s) \rangle  ds} \right)
   	(Q^{T}(T)X(T)dT +dM^{T}(T))\\[0.8em]
   	\quad \quad \quad \quad \quad+ \langle J(T),X(T) \rangle \exp \left(\int^{T}_t
   	{\langle J(s),X(s) \rangle  ds} \right)X(T)dT\\[0.8em]
   	\quad \quad \quad \quad=\exp \left(\int^{T}_t {\langle J(s),X(s) \rangle  ds}
   	\right)dM^{T}(T) +\langle J(T),X(T) \rangle Z(t,T)dT\\[0.8em]
   	\quad \quad \quad \quad \quad+ \exp \left(\int^{T}_t {\langle J(s),X(s)\rangle ds} \right)Q^{T}(T)X(T) dT \\[0.8em]
   	\quad \quad \quad \quad=\exp \left(\int^{T}_t {\langle J(s),X(s) \rangle  ds} \right)dM^{T}(T)\\[0.8em]
   	\quad \quad \quad \quad \quad + \left( Q^{T}(T)+ \operatorname{diag}[J(T)]
   	\right) Z(t,T)dT \\[0.8em]
   	\end{array}
   	\end{equation}
   	Integrating both sides of (\ref{diffZ}) gives
   	\begin{equation}\label{integdiffZ}
   	\begin{array}{ll}
   	\int_t^T dZ(t,s)= \int_t^T (Q^{T}(s) + \operatorname{diag}(J(s)))Z(t,s) ds \\[0.8em]
   	\quad \quad \quad \quad \quad \quad+ \int_t^T \exp \left(\int^{s}_t {\langle J(w),X(w) \rangle  dw} \right)dM^{T}(s).
   	\end{array}
   	\end{equation}
   	Put $\Psi(t,T;J)=\mathbb{E}^{T} \left[Z(t,T) \Big|\mathscr F_1(t)
   	\vee \mathscr F_2(t) \vee \mathscr F_3(t) \vee \mathscr F_X(t) \right]$.
   	Taking expectations in both sides of (\ref{integdiffZ}) results in
   	\begin{equation}\label{psi}
   	\Psi(t,T;J)=X(t) + \int_t^T (Q^{T}(s) +  \operatorname{diag}(J(s)))\Psi(t,s;J) ds.
   	\end{equation}
   	Suppose $\Phi(t,s;J)$ is the $N \times N$ matrix solution to the linear
   	system of ordinary differential equation
   	\begin{equation}
   	\left \{
   	\begin{array}{ll}
   	\dfrac{d\Phi(t,s;J)}{ds}=(Q^{T}(s)+\operatorname{diag}(J(s))) \Phi(t,s;J),\\[0.8em]
   	\Phi(t,t;J)=\operatorname{diag}(\mathbf{1})=I_N.
   	\end{array}
   	\right.
   	\end{equation}
   	Comparing with (\ref{psi}), we obtain the result $\Psi(t,T;J)=\Phi(t,T;J)X(t)$,
   	which finally gives us formula (\ref{eqn:tforward}).
   \end{proof}
   
   Now, substituting the result in Proposition 2 into (\ref{core}) gives us the
   characteristic function of  $y(T)=\ln S(T+\Delta)-\ln S(T)$ for the Heston-CIR
   model with regime-switching.
   
   \begin{prop}
   	If the underlying asset follows the dynamics (\ref{ForwardMeasureSDE}), then
   	the forward characteristic function of  $y(T)=\ln S(T+\Delta)-\ln S(T)$ is
   	given by
   	\begin{equation}
   	\begin{array}{ll}
   	f(\phi; t,T,\Delta,\nu(t),r(t))=\mathbb{E}^{T}[e^{\phi y(T)}|\mathscr F_1(t)
   	\vee \mathscr F_2(t) \vee \mathscr F_3(t) \vee \mathscr F_X(t)]\\[0.8em]
   	\quad \quad \quad \quad \quad \quad \quad \quad \quad=\exp (\nu(t)
   	G(D(\phi,T),t))\times \exp(r(t) M(E(\phi,T),t))\\[0.8em]
   	\quad \quad \quad \quad \quad \quad \quad \quad \quad \quad \times \langle
   	\Phi(t,T+\Delta;J) X(t),\mathbf{1} \rangle,
   	\end{array}
   	\end{equation}
   	where $D(\phi,t)$, $G(\phi,t)$, $J(t)$ and $\Phi(t,T+\Delta; J)$ are given by
   	\begin{equation}
   	\begin{array}{ll}
   	D(\phi,t)=\dfrac{a+b}{\sigma^2}\dfrac{1-e^{b(T+\Delta-t)}}
   	{1-ge^{b(T+\Delta-t)}},\quad  a=\kappa^*-\rho\sigma\phi, \quad b=\sqrt{a^2+\sigma^2(\phi-\phi^2)},\\[0.8em]
   	g=\dfrac{a+b}{a-b},\quad G(\phi,t)=\dfrac{2 \kappa^{*}\phi}{\sigma^2 \phi+(2\kappa^{*}-\sigma^2\phi)e^{\kappa^{*}(T-t)}},\\[0.8em]
   	J(t)=\left(\kappa^{*}\theta^{*}G(D(\phi,T),t)+ \alpha^{*}\beta^{*}M(E(\phi,T),t)
   	\right)(1-H_{T}(t))\\[0.8em]
   	\quad \quad \quad + \left(\alpha^{*}\beta^{*}E(\phi,t) +\kappa^{*}\theta^{*}
   	D(\phi,t)\right) H_{T}(t),\\[0.8em]
   	\Phi(t,T + \Delta;J)=\exp \left(\int_t^{T+\Delta} (Q^{T}(s) +
   	\operatorname{diag}(J(s))) ds \right),
   	\end{array}
   	\end{equation}
   	and $E(\phi,t)$ along with $M(\phi,t)$ are determined by the following ODEs
   	\begin{equation}
   	\left\{
   	\begin{array}{ll}
   	-\dfrac{dE}{dt}=\dfrac{1}{2} \eta^2 E^2 -(\alpha^*+B(t,T)\eta^2) E +\phi, \\[0.8em]
   	-\dfrac{dM}{dt}=\dfrac{1}{2} \eta^2 M^2 -(\alpha^*+B(t,T)\eta^2)M.\\[0.8em]
   	\end{array}
   	\right.
   	\end{equation}
   \end{prop}
   
   Now, by using the valuation of the fair delivery price for a variance swap,
   and summarizing the whole previous procedure, we can write the forward
   characteristic function for a variance swap as
   \begin{equation}
   \begin{array}{ll}
   \mathbb{E}^{T}\left(\left(\frac{S(t_{j})}{S(t_{j-1})}-1 \right)^2 \Big|
   \mathscr F_1(0) \vee \mathscr F_2(0) \vee \mathscr F_3(0) \vee
   \mathscr F_X(0) \right)\\[0.8em]
   =\mathbb{E}^T\left(e^{2y(t_{j-1})}-2e^{y(t_{j-1})}+1 |\mathscr F_1(0) \vee
   \mathscr F_2(0) \vee \mathscr F_3(0) \vee \mathscr F_X(0)\right)\\[0.8em]
   =f(2;0,t_{j-1},\Delta t,\nu(0),r(0))-2f(1;0,t_{j-1},\Delta t, \nu(0),r(0))+1,
   \end{array}
   \end{equation}
   where $y(t_{j-1})=\ln S(t_{j})-\ln S(t_{j-1})$, $\Delta t=t_{j}-t_{j-1}$, and
   the characteristic function $f(\phi; t,T,\Delta,\nu(t),r(t))$ is given in
   equation (\ref{core}). Hence, the fair strike price for a variance swap in terms
   of the spot variance $\nu(0)$ and the spot interest rate $r(0)$ under $T$-forward
   measure is given as
   \begin{eqnarray}\label{finalformchap5}
   K&=&\mathbb{E}^{T}(RV)  \\
   &=&\frac{100^2}{T}\sum_{j=1}^N \left(f(2;0,t_{j-1},\Delta t, \nu(0), r(0))
   -2f(1;0,t_{j-1},\Delta t, \nu(0), r(0))+1\right). \nonumber
   \end{eqnarray}
   
   \section{Formula validation and results}
   
   In this section, we assess the performance of formula \eqref{finalformchap5},
   by considering three regimes, denoted as $\{e_1, e_2, e_3\}$, representing
   the states $\emph{contraction}$, $\emph{trough}$ and $\emph{expansion}$ of
   the business cycle, respectively. The contraction state can be defined as
   the situation when the economy starts slowing down, whereas the trough state
   happens when the economy hits bottom, usually in a recession. In addition,
   expansion is identified as the situation when the economy starts growing
   again. Here, we assume that the Heston-CIR model without regime-switching
   corresponds to the first regime and it will switch to the other two regimes
   over time. Table 1 shows the set of parameters that we use to implement
   all the numerical experiments, unless otherwise stated.
   \begin{table}[htb]
   	\caption{\label{ParaHeston} Model parameters of the Heston-CIR hybrid model
   		with regime-switching.}
   	\renewcommand\tabcolsep{1pt}
   	\begin{tabularx}{\linewidth}{@{\extracolsep{\fill}}cc|cccc|ccc c|c}
   		\hline
   		$S_0$ &  $\rho $ & $V_0$ & $\theta^*$  & $\kappa^*$ & $\sigma  $ & $r_0$
   		& $\alpha^*$ & $\beta^*$ &  $\eta$ &  $T$ \\
   		\hline
   		1 & -0.4  & 0.05 & ${(0.05,0.075,0.04)}^{\intercal}$ & 2  &  0.1  & 0.05 &
   		1.2 & ${(0.05,0.04,0.075)}^{\intercal}$ & 0.01 & 1 \\
   		\hline
   	\end{tabularx}
   \end{table}
   
   \noindent
   In addition, the rate matrix for the Markov chain $\bf X$ is given by
   \begin{equation*}
   Q=\left(\begin{array}{ccc} -1 & 0.1 & 0.9 \\ 0.9 & -1 & 0.1 \\ 0.5 & 0.5
   & -1 \end{array} \right).
   \end{equation*}
   
   \subsection{Validation of the pricing formula against Monte Carlo simulation}
   
   We first demonstrate the validation of formula \eqref{finalformchap5} against
   Monte Carlo simulation. Here, the sampling frequency varies from $N=1$ up
   to $N=52$, and the Monte Carlo simulation is conducted using the Euler
   discretization with 200,000 sample paths. The comparison is displayed in
   Figure 1.
   
   \begin{figure}[!htb]
   	\caption{Strike prices of variance swaps for the Heston-CIR model with
   		regime-switching and the Monte Carlo simulation.}
   	{\includegraphics[width=350pt]{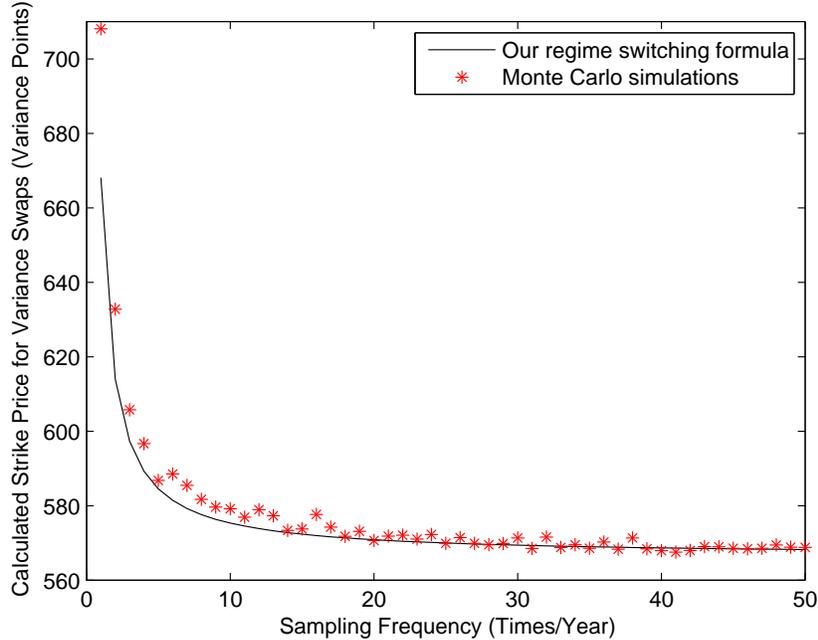}}
   \end{figure}
   
   As shown in Figure 1, our pricing formula provides a satisfactory fit to
   the simulation for $N=52$ which is the weekly sampling. In fact, the error
   calculated between our pricing formula and the simulation is less than
   0.077\% for $N=52$, and this error will be reduced as the number of sample
   paths increases. In addition, it should be emphasized that for $N=4$, the
   run time of our pricing formula is only 3.28 seconds, whereas the simulation
   takes about 8200 seconds. It is clear that our pricing formula attains almost
   the same accuracy in far less time compared to the simulation which serves
   as benchmark values.
   
   \subsection{Effect of regime-switching}
   
   In order to explore the effect of regime-switching, in Figure 2, we present
   results produced by  formula (\ref{finalformchap5}) and by the Heston-CIR
   model without regime-switching in \cite{Cao-Lian-Roslan}. For the Heston-CIR
   model without regime-switching, we fix the parameter values to be
   $\theta^*_1=0.05$ and $\beta^*_1=0.05$.
   \begin{figure}[!htb]
   	\caption{ Strike prices of variance swaps for the Heston-CIR model with and
   		without regime-switching.}
   	{\includegraphics[width=350pt]{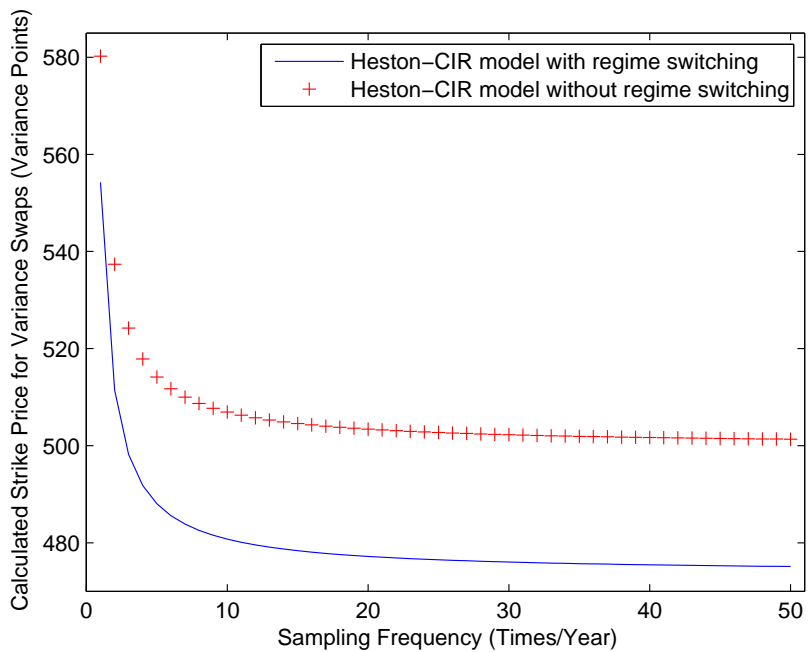}}
   \end{figure}
   
   We observe that the prices of variance swaps obtained from the Heston-CIR model
   with regime-switching are significantly lower than those from the corresponding
   model without regime-switching. For example, for $N=52$, the difference between
   variance swaps prices calculated from the two models is 7.32\%. This can be
   explained from the values of $\theta^*_1$ and $\beta^*_1$ which remain constant,
   whereas the values of $\theta^*$ and $\beta^*$ in the Heston-CIR model with
   regime-switching vary according to the changing states. Besides that, for the
   weekly sampling case, the difference in variance swaps prices between the two
   models becomes larger and stabilizes as the sampling frequency reaches 52. One
   possible explanation for this is the number of transitions between states in
   the Heston-CIR model with regime-switching increases as the sampling frequency
   increases.
   
   \medskip
   In addition, we also examine the economic aftermath for the prices of variance
   swaps by allowing the Heston-CIR model to switch across three regimes. In
   particular, we denote $\theta^*_1=0.05$ and $\beta^*_1=0.05$ for the contraction
   state, $\theta^*_2=0.075$ and $\beta^*_2=0.04$ for the trough state, and
   $\theta^*_3=0.04$ and $\beta^*_3=0.075$ for the expansion state, respectively.
   These values are assumed by noting that a good (resp. bad) economy is
   identified by high (resp. low) interest rate and low (resp. high) volatility.
   We provide the variance swaps pricing outcome for these three regimes in Table 2.
   \begin{table}[htb]
   	\caption{\label{compareswitching} Comparison of variance swaps prices among
   		different states in our pricing formula.}
   	\renewcommand\tabcolsep{1pt}
   	\begin{tabularx}{\linewidth}{@{\extracolsep{\fill}}c|c|c|c}
   		\hline
   		Sampling Frequency & State $\emph{Contraction}$ \quad & State $\emph{Trough}
   		\quad $ & State $\emph{Expansion} \quad $ \\
   		\hline
   		N=4  & 517.89  &  661.93 & 464.79 \\
   		N=12  & 505.74  & 648.32  & 450.21 \\
   		N=26  & 502.61  & 644.83 & 446.42 \\
   		N=52  &  501.28 & 643.37 & 444.82 \\
   		\hline
   	\end{tabularx}
   \end{table}
   
   From Table 2, we discover that the the price of a variance swap is highest in
   the trough state, followed by the contraction state, and found lowest in the
   expansion state. This trend is consistent throughout all sampling frequencies
   from $N=4$ to $N=52$. We can relate this finding to the economic condition of
   each of the states. In particular, the trough state is the state with the worst
   economy among the three, whereas the expansion state resembles the best economy.
   Thus, the price of a variance swap is cheapest in the best economy among the
   three, and most expensive in the worst economy among all. This implies that
   regime-switching has an important impact in capturing the economic changes on
   the prices of variance swaps.
   
   \section{Conclusion}
   
   The evaluation of variance swaps has been an active research topic in recent years. In \cite{shen2013pvs}, the continuously sampled variance swaps were
   priced under the regime switching Sch{\"o}bel-Zhu-Hull-White hybrid model. 
   However, variance swaps are  written on the realized variance  based on 
   daily closing prices in practice. To improve the pricing accuracy of these contracts, Zhu and Lian \cite{zhu-lian:11, zhu-lian:12} developed 
   closed-form pricing formulas of discretely sampled variance swaps based 
   on the framework of Heston's stochastic volatility model where the 
   interest rate followed a deterministic process. In  \cite{Cao-Lian-Roslan},  
   a hybridization of the Heston stochastic volatility model and the CIR 
   stochastic interest rate model was considered. The hybrid model extended 
   the Heston stochastic volatility model in \cite{zhu-lian:11} by modelling 
   the interest rate as the CIR process. The effect of stochastic interest
   rate on the price of discretely sampled variance swaps was demonstrated. 
   Elliott and Lian \cite{Elliott2012pva} made another extension of the 
   framework of Heston's stochastic volatility model in the direction of  
   including regime switching dynamics in the model. It was shown that
   incorporating regime switching into the Heston model had a significant 
   impact on the price of volatility swaps. 
   
   \medskip
   Since both regime switching and stochastic interest rate process affect 
   the price of variance swaps, we propose a model incorporating both 
   stochastic interest rate and regime switching effects. Specifically, 
   the proposed model combines the CIR stochastic interest rate into the Markov-modulated regime switching version of the Heston stochastic 
   volatility. Our model is capable of capturing several macroeconomic 
   issues such as alternating business cycles. In particular, we assume 
   that the long-term mean of variance of the risky stock and the 
   long-term mean of the interest rate depend on the states of the
   economy indicated by a regime-switching Markov chain. We demonstrate 
   our solution techniques and derive a semi-closed form formula for 
   pricing variance swaps. Numerical experiments reveal that our pricing 
   formula attains almost the same accuracy in far less time compared 
   with the MC simulation. To analyse the effects of incorporating
   regime-switching into pricing variance swaps, we first compare the 
   variance swaps prices calculated from the regime-switching Heston-CIR 
   model with the corresponding model without regime-switching. We find 
   that the prices of variance swaps obtained from the regime-switching 
   Heston-CIR model are significantly different from those from the 
   Heston-CIR model without regime-switching. In our case, the 
   Heston-CIR model without regime-switching corresponds to the state
   \emph{contraction}, and the price of a variance swap obtained from 
   the regime-switching Heston-CIR model is much lower than that 
   obtained from the Heston-CIR model without regime-switching. If 
   the Heston-CIR model without regime-switching corresponds to other 
   states, the conclusion can be different. Next, we explore the 
   economic consequence for the prices of variance swaps by allowing 
   the Heston-CIR model to switch across three regimes defined as 
   the best, moderate and worst economy. We notice that the price of 
   a variance swap is cheapest in the best economy among the three, 
   and most expensive in the worst economy among all. This confirms 
   the essence of incorporating regime-switching in pricing variance 
   swaps.


  \end{document}